\documentclass[a4paper,USenglish,cleveref, autoref,numberwithinsect]{lipics-v2019}

\title{The minimum cost query problem on matroids with uncertainty areas} 

\author{Arturo I. Merino}{Department of Mathematical Engineering and CMM, Universidad de Chile \& UMI-CNRS 2807, Santiago, Chile}{amerino@dim.uchile.cl}{0000-0002-1728-6936}{}

\author{José A. Soto}{Department of Mathematical Engineering and CMM, Universidad de Chile \& UMI-CNRS 2807, Santiago, Chile}{jsoto@dim.uchile.cl}{0000-0003-2219-8401}{}

\authorrunning{A.\,I. Merino and J.\,A. Soto}

\Copyright{Arturo I. Merino and José A. Soto}



\relatedversion{}

\supplement{}

\funding{Work supported by Conicyt via Fondecyt	 1181180 and PIA AFB-170001.}

\nolinenumbers 

\hideLIPIcs  

\EventEditors{Christel Baier, Ioannis Chatzigiannakis, Paola Flocchini, and Stefano Leonardi}
\EventNoEds{4}
\EventLongTitle{46th International Colloquium on Automata, Languages, and Programming (ICALP 2019)}
\EventShortTitle{ICALP 2019}
\EventAcronym{ICALP}
\EventYear{2019}
\EventDate{July 9--12, 2019}
\EventLocation{Patras, Greece}
\EventLogo{eatcs}
\SeriesVolume{132}
\ArticleNo{78}

\usepackage[textsize=tiny,textwidth=2cm]{todonotes}
\usepackage{algorithm}
\usepackage[noend]{algpseudocode}
\newcommand{\jscom}[1]{}
\newcommand{\amcom}[1]{}
\begin{document}

\newcommand{\spn}{\operatorname{span}}
\newcommand{\cospn}{\operatorname{cospan}}
\newcommand{\cl}{\operatorname{cl}}
\newcommand{\conv}{\operatorname{conv}}
\newcommand{\argmax}{\operatorname{argmax}}
\newcommand{\A}{\mathcal A}
\newcommand{\B}{\mathcal B}
\newcommand{\C}{\mathfrak C}
\newcommand{\R}{\mathcal R}
\newcommand{\M}{\mathcal M}
\newcommand{\I}{\mathcal I}
\newcommand{\Cofeas}{\texttt{Cofeas}}
\newcommand{\lo}{\texttt{low}}
\newcommand{\mi}{\texttt{mid}}
\newcommand{\hi}{\texttt{high}}
\newcommand{\bo}{\texttt{both}}
\newcommand{\core}{\texttt{core}}	
\newcommand{\mat}{\texttt{mat}}

\maketitle

\begin{abstract}
We study the minimum weight basis problem on matroid when elements' weights are uncertain. For each element we only know a set of possible values (an uncertainty area) that contains its real weight. In some cases there exist bases that are uniformly optimal, that is, they are minimum weight bases for every possible weight function obeying the uncertainty areas. In other cases, computing such a basis is not possible unless we perform some queries for the exact value of some elements. 

Our main result is a polynomial time algorithm for the following problem. Given a matroid with uncertainty areas and a query cost function on its elements, find the set of elements of minimum total cost that we need to simultaneously query such that, no matter their revelation, the resulting instance admits a uniformly optimal base. We also provide combinatorial characterizations of all uniformly optimal bases, when one exists; and of all sets of queries that can be performed so that after revealing the corresponding weights the resulting instance admits a uniformly optimal base.

\end{abstract}
\section{Introduction}

We study fundamental combinatorial optimization problems on weighted structures where the numerical data is uncertain but it can be queried at a cost. We focus on the problem of finding a minimum weight base of a matroid under uncertainty, a problem that includes finding the smallest $k$ elements of a list and the minimum spanning tree (MST) problem. In our setting, for every element $e$ of the matroid we know a set $\A(e)$, called uncertainty area, of possible values that contains its real weight $w_e$. We can reveal this real weight by paying some query cost $c_e$. We assume that the queries are done in a non-adaptive way, or equivalently, that all the elements queried reveal their values at the same time. A set of elements $F$ is a \emph{feasible query} if for every possible revelation of the weights of $F$ it is possible to compute a minimum weight base $T$ of the resulting instance. The task of the \emph{Minimum Cost Query Problem on Matroids} is to determine a minimum-cost feasible query.

\begin{figure}
\tikzstyle{vertex}=[circle,fill=black!25,minimum size=10pt,inner sep=0pt]
\tikzstyle{edge} = [draw,thick,-]
\tikzstyle{weight} = [font=\small]
\begin{subfigure}[t]{0.3\textwidth}
\centering
\begin{tikzpicture}[scale=0.8, auto,swap]
\node[vertex] (a) at (0,2) {};
\node[vertex] (b) at (2,1) {};
\node[vertex] (c) at (2,3) {};
\path[edge] (a) -- node[weight] {$[0,2]$} (b);
\path[edge] (b) -- node[weight] {$[4,9]$} (c);
\path[edge] (c) -- node[weight] {$[1,3]$} (a);
\end{tikzpicture}
\caption{The empty set is feasible.}
\end{subfigure}\hfill
\begin{subfigure}[t]{0.3\textwidth}
\centering
\begin{tikzpicture}[scale=0.8, auto,swap]
\node[vertex] (a) at (0,2) {};
\node[vertex] (b) at (2,1) {};
\node[vertex] (c) at (2,3) {};
\path[edge] (a) -- node[weight] {$[0,1]$} (b);
\path[edge] (b) -- node[weight] {$[0,1]$} (c);
\path[edge] (c) -- node[weight] {$[0,1]$} (a);
\end{tikzpicture}
\caption{Only the full set is feasible.}
\end{subfigure}\hfill
\begin{subfigure}[t]{0.3\textwidth}
\centering
\begin{tikzpicture}[scale=0.8, auto,swap]
\node[vertex] (a) at (0,2) {};
\node[vertex] (b) at (2,1) {};
\node[vertex] (c) at (2,3) {};
\path[edge] (a) -- node[weight] {$\{0,1\}$} (b);
\path[edge] (b) -- node[weight] {$\{0,1\}$} (c);
\path[edge] (c) -- node[weight] {$\{0,1\}$} (a);
\end{tikzpicture}
\caption{Any 2-set is feasible.}
\end{subfigure}
\caption{Feasible queries.   }\label{fig:example}
\end{figure}
To better illustrate this, consider the problem of computing an MST of a triangle graph in three possible situations as shown in Figure 1. In the first situation, the edges with areas $[0,2]$ and $[1,3]$ always form an MST, so we don't need to query any element. In this case we say that the matroid (the graph) admits a \emph{uniformly optimal basis}, that is, a basis (a spanning tree) having minimum weight for every possible realization of the elements' weights. In the second situation, all edges have uncertainty area $[0,1]$ and the only feasible query is the entire set of edges. For if we only query two edges, we could be in a situation where both have weight $1/2$. With that information we cannot compute an MST: if the unqueried edge $e$ had weight 0, then $e$ must be in the  MST. However, if it had weight 1, then $e$ cannot be in an MST. In the last situation, all uncertainty areas are finite: they have two elements $\{0,1\}$. Here, every set of size two is feasible. Indeed, if both elements reveal a weight of $0$, then they form an MST. Otherwise, the set obtained by deleting any edge with weight $1$ is an MST.

\paragraph*{Paper outline and results} In Section \ref{sec:Preliminaries} we give formal definitions and show a simple but strong result: the uniformly optimal bases (UOB) of an uncertainty matroid form the bases of a second matroid. In Section \ref{sec:Colors} we extend the classic MST red and blue rules to uncertainty matroids, introducing the concept of \emph{colored} elements, and study their properties. We then show that an uncertainty matroid admits a UOB if and only if all its elements are colored, and use this to give a combinatorial description of the matroid of uniformly optimal bases. We provide polynomial time algorithms for testing the existence of a UOB and for finding one if it exists. 
In Section \ref{sec:Feasible} we study the minimum cost feasible query problem in detail. By using our coloring framework we construct a partition of the elements into groups that characterize minimal feasible queries. We show that every minimal feasible query set is formed by taking the first group (denoted as the core) completely and by deleting exactly one element from each other group. Our main result is an algorithm to find this partition that we  use to fully solve the minimum cost feasible query problem on matroids. Unlike related work on the MST, the uncertainty areas in our setting can be arbitrary sets of real numbers, and not just intervals. Our algorithms only assume access to an independence oracle for the matroid, and a very mild type of access to the uncertainty areas. At the end of this section we show that the interval uncertainty area case is special, as there is a unique minimal size feasible query. We also relate the solutions for the MST with $\{0,1\}$-areas case with the 2-connected components of the graph. Finally, in Section \ref{sec:algorithm} we discuss implementation details of our algorithms.

\paragraph*{Related Work}

Traditional research in optimization with uncertain data mostly focuses on finding solutions whose value is good, either in the worst case (robust optimization) or in some probabilistic sense (stochastic optimization), without gaining new information about the uncertain data. Our problem contributes to the \emph{query model setting}, a different approach that has gained some strength in the last years. In this model one assumes that the algorithm can learn the exact value of an uncertain input data by paying some query cost in order to solve a certain problem P (e.g., determining  the MST of a graph). The aim is to minimize the cost of the queries while guaranteeing that an exact/approximate solution can be computed. Work in this area (see \cite{ErlebachH15} for a survey)  falls into three main categories:

\subparagraph*{Adaptive online} In this setting, the algorithm can query elements one by one, using the information revealed until a certain step to guide the decision of the next element to query. Even though algorithms for this version can be analyzed with a traditional worst-case approach (e.g., minimizing the depth of the decision tree associated to the algorithm's strategy), most of the work in this setting prefers to measure  performance in terms of competitive analysis, comparing the number of queries an algorithm makes with the minimum number of queries that an adversary algorithm (that knows the real values beforehand) would make in order to verify that an answer is correct. Probably the first one to consider this model is Kahan \cite{Kahan91} who provide optimal online algorithms to compute \emph{all the elements} achieving the maximum, median and minimum gap of a list of closed intervals. Feder et al. \cite{FederMPOW03} devise optimal online competitive algorithms to determine the numerical \emph{value} of the median, and more generally, the $K$-th top element of the list within some prescribed tolerance. Bruce et al. \cite{BruceHKR05} introduce a general method, called the \emph{witness algorithm}, for adaptive online problems with open intervals and singletons uncertainty areas. They apply this method to geometric problems such as finding all maximum points in the plane from a family with uncertain coordinates. Erlebach et al. \cite{ErlebachHKMR08} studied the MST problem under two types of uncertainty, the edge uncertainty one (which is the same as ours) and the vertex uncertainty setting, in which the graph is complete, the vertices are points in the plane whose coordinates are uncertain, and the weight of an edge is the distance between its endpoints. They get 2 and 4 competitive algorithms, respectively, for both types of uncertainty, under open intervals and singletons areas, which is optimal for deterministic algorithms. The algorithm for edge uncertainty, denoted as U-RED, was later extended by Erlebach, Hoffmann and Kammer \cite{Erlebach0K16} to the minimum weight base problem on matroids achieving an optimal 2 competitive algorithm. Megow, Mei{\ss}ner and Skutella \cite{MegowMS17} show that by using randomization one can do better, lowering the competitive ratio down to $1+1/\sqrt{2}$. They also studied the non-uniform cost case. Gupta et al. \cite{GuptaSS16} studied variants where queries return refined 
estimates of the areas, instead of a single value.
\subparagraph*{Verification} The verification problem is the one the offline adversary of the previous setting has to solve. That is, given both the uncertainty areas and a family of assumed real values, to determine the minimum number of queries one has to make so that, if the values obtained from the queries and the assumed values coincide, then no more queries are needed in order to obtain an optimal solution. Charalambous and Hoffman \cite{CharalambousH13} show that the verification problem for  maximal points in the plane is NP-hard for uncertainty areas of size at most 2. Erlebach and Hoffmann  \cite{ErlebachH15} show that the verification problem of MST with (open interval and singleton) uncertainty in the edges is in P, while that for vertex uncertainty is NP-hard.

\subparagraph*{Non-adaptive online.} This setting encompasses our work and is sometimes called the \emph{offline} problem. In it, an algorithm must determine a set $F$ of queries to perform simultaneously, in order to have enough information to solve the problem. The only conceptual difference between the non-adaptive online problem and the verification one, is that in the latter, the algorithm can make use of the real values of the elements to guide the decision of which queries to perform, while in the former, that information is not available. Feder et al. \cite{FederMPOW03} provide optimal algorithms for finding the $K$-th top element of a list up to  additive tolerance. Later, Feder et al. \cite{FederMOOP07} consider the problem of finding the shortest $s$-$t$ path on a DAG with closed-intervals uncertainty on the edges. They show that determining the length of the shortest $s$-$t$ path within a given additive error is neither in NP nor in co-NP unless NP=co-NP, and provide exact algorithms for some special cases. To the best of our knowledge, the MST and, more generally, the matroid case have not been considered before this work.

\subparagraph*{Further related work} A common approach to deal with closed-interval uncertainty for a problem, without querying extra information, is to find a solution that minimizes the maximal regret, that is, the difference between the (real) weight of the chosen solution and the weight of the best solution that could have been picked had the true weights been known. Zero maximal regret bases and uniformly optimal bases of a matroid coincide. In the context of the MST problem, Yaman et al. \cite{YamanKP01} characterize trees with zero regret when they exists. In \cite{KasperskiZ06}, Kaspersky and Zielinsky give a 2-approximation algorithm  for the minmax regret problem on general matroids with interval data, and in \cite{KasperskiZ07} they give algorithms to find zero maximal regret bases. It is worth noting that all these results assume interval areas on the elements, while we allow for arbitrary uncertainty areas. 

\section{Preliminaries. Uniformly Optimal Bases.} \label{sec:Preliminaries}
We assume familiarity with basic concepts in matroid theory such as bases, independent sets, span, circuit, cocircuits and duality, contraction, deletion and matroid connectivity. For an introduction and specific results, we refer to Oxley’s book \cite{Oxley}. However, most of this paper can be understood by having the graphic matroid of a connected graph in mind. This is a matroid whose elements are the edges. The bases, independent sets, circuits and cocircuits, are the spanning trees, forests, cycles and minimal edge cut-sets, respectively. An element $e$ is in the span (cospan) of a set $F$  if there is a circuit (cocircuit) in $F\cup \{e\}$ containing $e$. We also use the standard notation $X+e$ and $X-e$ to denote $X\cup \{e\}$ and $X\setminus \{e\}$.

\begin{definition}[Uncertainty Matroid] An uncertainty matroid is a pair $(\M,\A)$ where $\M=(E, \I)$ is a matroid and $\A\colon E \to 2^\mathbb R \setminus \{\emptyset\}$ is a function mapping each element $e\in E$ in the ground set to a nonempty set $\A(e)$ of real numbers denoted the uncertainty area of $e$. We denote $\inf \A(e)$ by $L_e$ and $\sup \A(e)$ by $U_e$.
\end{definition}

If $\A(e)$ is a singleton, we say that the element $e$ is \emph{certain}, otherwise, we say it is \emph{uncertain}. If all the elements are certain, then we can identify $\A$ with the associated weight function $w\colon E\to \mathbb R$, such that $\A(e)=\{w_e\}$ so that $(\M,w)$ becomes a weighted matroid.

\begin{remark}
For our algorithms, we assume access to matroid $\M$ via an independence oracle. We also need a very mild access to the uncertainty areas. More precisely, we assume that both $L_e$ and $U_e$ are known for every element, and that for every pair of (not necessarily distinct) elements $e$ and $f$  we can test if $(L_e,U_e)\cap \A(f)$ is empty in constant time. A polynomial time algorithm will, therefore, only use a polynomial number of calls to the independence oracle and to the uncertainty areas.
Furthermore, for simplicity we will assume that all the infima and suprema $L_e$ and $U_e$ are finite. Otherwise we can apply a suitable strictly increasing function mapping the reals to a bounded set (such as $\arctan$), and work in its image instead. 
\end{remark}

Intuitively, an uncertainty matroid models the situation in which we do not know the actual weight of every element $e$ in the matroid, but we do know a set $\A(e)$ containing it. We can learn the actual weight of $e$ by querying it. In this work we are concerned with non-adaptive (simultaneous) queries. The new uncertainty area function obtained after querying a subset of elements will be called a \emph{revelation} of $\A$. The formal definition is below.

\begin{definition}[Revelations and realizations] Let $X\subseteq E$. A revelation of $X$ in $(\M,\A)$ is a function $\B\colon E \to 2^\mathbb R\setminus \{\emptyset\}$ such that:
	\begin{romanenumerate}
		\item $\forall e \in X$, $\B(e)=\{b_e\}\subseteq \A(e)$ is a singleton, and
		\item $\forall e \in E\setminus X$, $\B(e)=\A(e)$.
	\end{romanenumerate}
	In particular, $(\M,\B)$ is also an uncertainty matroid. A \emph{realization} is a revelation of the entire ground set $E$. The collection of all revelations of $X$ in $(\M,\A)$ is denoted by $\R(X,\A)$.
\end{definition}

Suppose we want to compute a minimum weight basis of a matroid, but we only know uncertainty areas for its elements. In certain situations (e.g., Figure \ref{fig:example} (a)), we can find sets that are optimal bases for every realization, we call them uniformly optimal bases.

\begin{definition}[Uniformly optimal bases] A set $T\subseteq E$ is a uniformly optimal basis of $(\M,\A)$ (or simply, an $\A$-basis) if for every realization $w$, $T$ is a minimum weight basis (a $w$-basis).
\end{definition}
 Recall that a nonempty family of sets form the bases of a matroid if and only if they satisfy the strong basis exchange axiom.  Our first basic result is the following:

\begin{lemma}\label{lem:uob-matroid}
	Let $\mathfrak B$ be the collection of all uniformly optimal bases of $(\M,\A)$. Suppose $\mathfrak B\neq \emptyset$ and let $T_1$ and $T_2$ be two sets in $\mathfrak B$. If $e$ is an element in $T_1\setminus T_2$ then:
	\begin{romanenumerate}
		\item $e$ is certain, and
		\item strong basis exchange holds, i.e., there is an element $f$ in $T_2\setminus T_1$ such that both $T_1-e+f$ and $T_2-f+e$ are in $\mathfrak B$.
	\end{romanenumerate}
	In particular, if $\mathfrak B\neq \emptyset$, then $\mathfrak B$ is the set of bases of a matroid, that we denote by $\mat(\M,\A)$
\end{lemma}

\begin{proof}
	Let $e \in T_1 \setminus T_2$. By strong basis exchange of $\M$, there is an element $f\in T_2 \setminus T_1$ such that both $T_1-e+f$ and $T_2+e-f$ are bases of $\M$. Assume first by contradiction that $e$ is uncertain, then there is a realization $w\in \R(E,\A)$ such that $w_e \neq w_f$. If $w_e>w_f$ then $w(T_1-e+f)<w(T_1)$ contradicting the fact that $T_1$ is uniformly optimal. On the other hand, if $w_e < w_f$ then $w(T_2+e-f)<w(T_2)$ contradicting that $T_2$ is uniformly optimal. We conclude not only that $e$ is certain, but also that in every realization $w_f=w_e$. In particular, for every realization $w$, $w(T_1-e+f)=w(T_1)=w(T_2)=w(T_2-f+e)$, i.e. both $T_1-e+f$ and $T_2-f+e$ are uniformly optimal bases. \end{proof}

\section{Blue and red rules for uncertainty matroids.}\label{sec:Colors}
By Lemma \ref{lem:uob-matroid}, we conclude that if an uncertain element $e$ is in some uniformly optimal basis then it is in \emph{every} uniformly optimal basis. Our next task is to characterize the set of uncertain elements that are in every uniformly optimal basis. Now it is useful to remember the classic blue and red rules for computing an MST. An edge of a weighted graph is called blue if it is in at least one MST, and it is called red if it is outside at least one MST. Virtually every algorithm for the MST work in steps: if it detects a red edge from the current graph, it deletes it from the graph, and if it detects a blue edge then it adds it to the solution and contracts it. The following definitions extend the coloring notion to uncertainty matroids. 

\begin{definition}[Blue and red elements]\label{blueAndRed} An element  $e\in E$ is blue (resp. red) if for every realization $w\in \R(E,\A)$ there exists a $w$-basis $T$ such that $e\in T$ (resp. $e\notin T$). We say that $e$ is colored if $e$ is red or blue (or both at the same time), otherwise we say that $e$ is uncolored.\end{definition}

Note that $e$ is blue (resp. red) on $(\M,\A)$ if and only if $e$ is blue (resp. red) on $(\M, w)$ for every realization $w\in \R(E,\A)$. Standard matroid arguments show that for any such $w$, if an element $e$ is any heaviest element of a circuit, then it avoids some $w$-basis (i.e., it is red on $(\M,w)$), and if it is the unique heaviest element, then it cannot be in a $w$-basis (i.e., it is not blue on $(\M,w)$). Blue elements in graphs / matroids with interval uncertainty areas have been studied before under the name of strong edges \cite{YamanKP01} or necessarily optimal elements \cite{KasperskiZ07}.  We start with a basic result on colored elements, they are preserved under revelations. 

\begin{lemma}\label{colorPreservationRevelation}
	Let $X\subseteq E$ and $\B \in \R(X,\A)$. If $e$ is blue (resp. red) in $(\M,\A)$, then it is blue (resp. red) in $(\M,\B)$.
\end{lemma}

\begin{proof}
    Suppose $e$ is blue in $(\M,\A)$ and consider any $w\in \R(E,\B)$. Since $\R(E,\B)\subseteq \R(E,\A)$ we get that there exists some $w$-basis $T$ such that $e\in T$, implying that $e$ is blue in $(\M,\B)$. The proof works analogously if $e$ is red in $(\M,\A)$. 
\end{proof}

We can characterize blue and red elements of an uncertainty matroid using its span and cospan functions\footnote{We recall that $\spn(X)$ for $X\subseteq E$ is the unique maximal set $U\supseteq X$ with the same rank as $X$. The cospan of $X$, $\cospn(X)$ is its span in the dual matroid. The span of a matroid and its cospan are related by the expression $\cospn(X)=X \cup \{e\in E : e\notin \spn[(E-e)\setminus X] \}$} of a set $X$ in the matroid $\M$ (resp., in the dual matroid $\M^*$). The span and cospan of a set can be computed using a polynomial number of calls to the independence oracle.

\begin{lemma}\label{blueAndRedCharacterization}
	Let $F(e)=\{f \in E -e : L_f < U_e \}$ and $F^*(e)=\{f \in E-e : L_e < U_f \}$ .
	\begin{enumerate}[(i)]
	    \item \label{br1} $e \in E$ is blue if and only if $e\notin \spn F(e)$.
		\item \label{br2} $e \in E$ is red if and only if $e \notin \cospn F^*(e)$.
	\end{enumerate}
	In particular, we can test in polynomial time if any element is blue, red, both or none.
\end{lemma}

\begin{remark}
    Proofs involving blue and red elements follow the same main ideas. Most of the time we  only consider the blue case, since the red case follows by \emph{duality arguments}. More precisely, consider the  dual matroid $\M^*$ with inverted  uncertainty area  function $-\A(e)=\{-x \colon x\in \A(e)\}$. Since the bases of $\M^*$ are complements of the bases of $\M$ we also get that the uniformly optimal bases of $(\M,\A)$ are the complements of those in $(\M^*,-\A)$, that the red elements in $(\M,\A)$ are the blue elements of $(\M^*,-\A)$ and vice versa. 
\end{remark}

\begin{proof}[Proof of Lemma \ref{blueAndRedCharacterization}]

	We only prove \eqref{br1}, since \eqref{br2} follows by duality arguments using that $F^*(e)$ is the analogue of $F(e)$ for $(\M,-\A)$.
	
	Let $e$ be a blue element and let $K= \min\limits_{f \in F(e)} (U_e - L_f)>0$. For each $f\in E-e$  choose $\varepsilon_f\in [0,K/2)$ such that $L_f +\varepsilon_f \in \A(f)$. In a similar way, choose $\varepsilon_e\in [0,K/2) $ such that $U_e-\varepsilon_e \in \A(e)$. Consider the  realization $w \in \R(E,\A)$ given by $
	w_f = \begin{cases}   U_e - \varepsilon_e  & \text{ if $f=e$,} \\
	L_f + \varepsilon_f  & \text{ if $f\neq e$.}
	\end{cases}$
	
	By construction, for every $f \in F(e)$, $w_f<w_e$. 
	Suppose now that $e \in \spn F(e)$. Then, there exists a circuit $C \subseteq F(e)+e$ such that $e$ is the unique heaviest (with respect to $w$) element of $C$. This implies that $e$ is outside every $w$-basis, which contradicts that $e$ was blue.
	
	Now let $e$ be an element outside $\spn F(e)$ and suppose that $e$ is not blue. Then, there exists a realization $w \in \R(E,\A)$ such that $e$ is not in any $w$-basis. Choose $T$ to be any $w$-basis and let $C$ be the fundamental circuit\footnote{Recall that if $T$ is a basis and $e$ an element, the fundamental circuit $C$ of $T+e$ is the only circuit of $T+e$. For any $f\in C$, $T-f+e$ is a basis.} of $T+e$. As $e \notin \spn F(e) $ we have that $C-e \not\subseteq F(e)$. Select any $f\in (C-e)\setminus F(e)$, then:
	$w_f \geq L_f \geq U_e \geq w_e$. 
	We conclude that $T-f+e$ is a $w$-basis that contains $e$, which is a contradiction.
\end{proof}

Even though certain elements can be blue and red at the same time (e.g., in a circuit in which every element has the same weight every element has both colors), this is not possible for uncertain ones as shown by the next lemma.

\begin{lemma}\label{redBlueCertain}
    If $e$ is red and blue, then $e$ is certain.
\end{lemma}
\begin{proof}
Suppose for the sake of contradiction that $e$ is uncertain and pick $w \in \R(E,\A)$ such that $w_e < U_e$. As $e$ is red there exists some $w$-basis $T$ without $e$. Let $C$ be the fundamental circuit of $T+e$. Since $e$ is blue, there exists $f\in (C-e)\setminus F(e)$. Then $w_f \geq L_f \geq U_e > w_e,$ implying that $w(T-f+e)<w(T)$ which contradicts the fact that $T$ was a $w$-basis.
\end{proof}

The definition of blue and red will be useful even if $e$ is certain, or if the uncertainty matroid has no uniformly optimal basis. The following lemma highlights the utility of this definition.

\begin{lemma}\label{blueAndRedUncertain}
Let $(\M,\A)$ be an uncertainty matroid such that an $\A$-basis exists.
\begin{romanenumerate}	
\item If $e$ is blue (resp.\ red)  then there exists an $\A$-basis $T$ such that $e\in T$ (resp. $e\not\in T$).
\item Let $e$ be an uncertain element. Then $e$ is blue (resp. red) if and only if $e$ is inside (resp. outside) every $\A$-basis.
\end{romanenumerate}
\end{lemma}
\begin{proof}
We only prove the blue case for both items, as the red one follows by duality arguments. For (i), let $S$ be any $\A$-basis. If $e\not\in S$ then consider the fundamental circuit $C$ of $S+e$. Since $e$ is blue we have that $C-e\not\subseteq F(e)$. Let $f\in (C-e)\setminus F(e)$ and let $T=S-f+e$. We claim that $T$ is an $\A$-basis containing $e$,  otherwise there exists a realization $w\in \R(E,\A)$ such that $w(S)<w(T)=w(S)-w_f+w_e$, but then $L_f\leq w_f<w_e\leq U_e$, implying that $f\in F(e)$ which is a contradiction.
Now consider (ii). Note that if $e$ is blue then by (i), it is contained in some $\A$-basis. But since $e$ is uncertain we have, by Lemma  \ref{lem:uob-matroid}, that $e$ is in every $\A$-basis. For the converse, let $T$ be any $\A$-basis. Since $T$ is a $w$-basis for every realization $w$ then by definition, all its elements are blue.
\end{proof}
The previous lemma shows that if the set of uniformly optimal bases is nonempty then all uncertain elements are colored either red or blue. Furthermore the uncertain elements contained in every $\A$-basis are exactly the blue uncertain elements. The next theorem, which is the main result of this section, shows that a converse also holds.

\begin{theorem}\label{existanceCharacterization}
 An $\A$-basis exists if and only if every uncertain element is colored.
\end{theorem}

In order to prove this theorem we need to prove three simple lemmas. Basically, they show that contracting blue elements and/or deleting red elements preserves structure and colors. 

\begin{lemma}\label{blueAndRedSafety} Let $e\in E$ be a colored element (certain or uncertain).
\begin{romanenumerate}
\item If $e$ is blue:\quad 
		$T$ is an $(\M/e,\A|_{E-e})$-basis if and only if $T+e$ is an $(\M,\A)$-basis.
\item If $e$ is red:\quad $T$ is an $(\M\setminus e,\A|_{E-e})$-basis if and only if $T$ is an $(\M,\A)$-basis.
\end{romanenumerate}
\end{lemma}
\begin{proof}
We give proof only for the case when $e$ is blue; as the other case follows by duality arguments. For the direct implication observe that $T$ is a basis of $\M/e$ and $e$ is not a loop of $\M$, therefore $T+e$ is a basis of $\M$. Let us suppose that $T+e$ is not an $(\M,\A)$-basis, then there exists a basis $T'$ and a realization $w \in \R(E,\A)$ such that $w(T')<w(T+e)$. As $e$ is blue in $(\M,\A)$ there exists $T''$ an $(\M,w)-$basis such that $e \in T '' $, so $w(T'')=w(T')$. Notice that $T''-e$ is a basis of $\M/e$. Since $w(T''-e)<w(T)$ this contradicts the fact that $T$ is $(\M/e, \A|_{A-e})$-basis.
	
	For the converse, suppose that $T$ is not an $(\M/e,\A|_{A-e})$-basis. Then we have $T'$ a basis of $\M/e$ and a realization $w \in \R(E-e, \A|_{A-e})$ such that $w(T')<w(T)$. Extending $w$ to $\hat w \in \R(E,\A)$ by selecting any $\hat w_e \in \A(e)$ we have that $\hat w (T'+e)< \hat w(T+e)$. Since $T'+e$ is a basis of $\M$ we have contradicted that $T$ is an $(\M,\A)$-basis. 
\end{proof}

The second lemma we need allows us to simplify the uncertainty areas for coloring purposes.

\begin{definition}
	Let $(\M,\A)$ be an uncertainty matroid. We define the closure of $\A$ as:
	$$
	\cl \A = \{ [L_e,U_e] : e\in E   \}
	$$
\end{definition}
\begin{lemma}
	\label{closureLemma}
	Let $(\M,\A)$ be an uncertainty matroid. If $\B$ is an uncertainty area function  such that $L_e^\A=L_e^\B$ and $U_e^\A=U_e^\B$  for each $e\in E$, then $(\M,\A)$ and $(\M,\B)$ have the same colors. In particular $(\M,\A)$ and $(\M,\cl \A)$ have the same colors.
\end{lemma}
\begin{proof}
	Note that $F^\A (e) = F^{\B} (e)$ and ${F^*}^\A (e) = {F^*}^{\B} (e)$. By Lemma \ref{blueAndRedCharacterization} we have that $(\M,\A)$ and $(\M,\B)$ have the same colors.
\end{proof}
\begin{lemma}\label{colorPreservation} Let $e\in E$ be a colored uncertain element, and $f\in E- e$.
	\begin{romanenumerate}
		\item Suppose $e$ is blue. If $f$ is blue (resp.\ red) in $(\M,\A)$, then $f$ is blue (resp.\ red) in $(\M/ e, \A|_{E-e})$.
		\item Suppose $e$ is red. If $f$ is blue (resp. red) in $(\M,\A)$, then $f$ is blue (resp. red) in $(\M \setminus e, \A|_{E-e})$.
	\end{romanenumerate}
\end{lemma}
\begin{proof}
    We give proof only for the case when $e$ is blue; as the other case follows by arguing dually.  Define the weight function $w:E \to \mathbb R$ given by:
	\[
	w_x = \begin{cases}
		U_f & \text{if $x=f$,}\\
		L_x & \text{if $x\neq f$.}
	\end{cases} \qquad w^*_x = \begin{cases}
	L_f & \text{if $x=f$,}\\
	U_x & \text{if $x\neq f$.}
	\end{cases}
	\]
We begin by considering $f$ blue. Since $f$ is blue in $(\M,\cl \A)$, there exists some $(\M,w)$-basis $T$ such that $f \in T$. We start by proving that $e\in T$, suppose not, then we can consider $C$ the fundamental circuit of $T+e$. As $e$ is blue, we have that $e \notin \spn F(e)$ and $(C-e) \not\subseteq F(e)$. Hence, we can select $g \in C \setminus  F(e)$ such that $g\neq e$. Since $e$ is non-trivial: 
	\[w_g \geq L_g \geq U_e > L_e = w_e.\]  
	It follows that $w(T-g+e)<w(T)$, which contradicts the fact that $T$ is a $w$-basis. By Lemma \ref{blueAndRedSafety}, we have that $T-e$ is a $(\M/e, w|_{E-e})$-basis such that $f \in T-e$. Therefore, $f\notin \spn_{\M/e} \{g \in E-e : w_g < w_f \}=\spn_{\M/e} (F(f)-e)$. Using Lemma \ref{blueAndRedCharacterization} it follows that $f$ is blue in $(\M/e, \A|_{A-e})$. 

	 We now turn to the case $f$ red. As $f$ is red in $(\M,\cl \A)$, there exists some $(\M,w^*)$-basis $T$ such that $f \notin T$. If $e\in T$ then $T-e$ is an $(\M/e,w^*|_{E-e})$-basis by Lemma \ref{blueAndRedSafety}. As $f$ is in a maximum weight cobasis in $\M/e$, we get \[f\notin \cospn_{\M/e} \{g \in E-e : w^*_g > w^*_f \}=\cospn_{\M/e} (F^*(f)-e).\] By Lemma \ref{blueAndRedCharacterization} we conclude that $f$ is red in $(\M,\A)$. Hence, we are only left with the case $e\notin T$. Since $e$ is blue in $(\M,\cl \A)$, there exists some $(\M,w^*)$-basis $T'$ such that $e\in T'$. As $(\M,w^*)$-basis are basis of a matroid (because they are basis of minimum weight with respect to $w^*$) we can find $e ' \in T\setminus T'$  such that $T-e'+e$ is an $(\M,w^*)$-basis, from here we can proceed as before but working with $T-e'+e$ instead of $T$.
\end{proof}

We are now ready to prove this section's main result.
\begin{proof}[Proof of Theorem \ref{existanceCharacterization}]
	We only need to prove the converse and we proceed by induction on the number of uncertain elements $k$.
	If $k=0$ an $\A$-basis is simply a basis of minimum weight, which clearly exists. Suppose that $k>0$, and let $e\in E$ be any uncertain element. If $e$ is blue, we have that every uncertain element of $E-e$ is colored in $(\M/e, \A|_{E-e})$, as colors were preserved. By inductive hypothesis, we have an $(\M/e,\A|_{E-e})$-basis $T$ and by Lemma \ref{blueAndRedSafety} we have that $T+e$ is an $(\M,\A)$-basis. If $e$ is red, one can proceed similarly but deleting instead.
\end{proof}

\begin{remark}
The previous theorem gives an algorithmic way to test if $(\M,\A)$ admits an $\A$-basis: we simply check if every element is colored, using Lemma \ref{blueAndRedCharacterization}. Let us now consider the problem of finding one such base. It is worth noting that algorithms for this task are available for closed interval and open interval uncertainty areas. For closed interval uncertainty areas, one can find an $\A$-basis by using Kaspersky and Zielinsky's approach \cite{KasperskiZ06}  for finding a zero maximal regret basis. For open intervals (or singletons), one can simply run the 2-competitive U-RED algorithm by Erlebach, Hoffmann and Kammer \cite{Erlebach0K16} for the online adaptive variant: if this algorithm does not perform any query, then it  outputs a uniformly optimal basis. Otherwise, the algorithm finds a witness set, i.e., a set for which at least one element must be queried in order to find a solution.
In what follows we  provide a new algorithm that finds $\A$-bases for arbitrary  uncertainty areas (not just intervals).
\end{remark}

If uniformly optimal bases exist, then the elements are partitioned into blue uncertain, red uncertain and certain elements. After contraction of the set $B$ of all blue uncertain elements and deletion of the set $R$ of all red uncertain elements, we are left with a matroid that only has certain elements, we call such weighted matroid the \emph{certain weighted matroid} $(\M^c,w^c)$, where $\M^c=\M/B\setminus R$, and $w^c=\A|_{E\setminus (R\cup B)}$. The following theorem shows that every uniformly optimal basis arises by extending some  optimal basis on the certain weighted matroid.

\begin{theorem}
	\label{UniformBasesCharacterization}
	Let $(\M,\A)$ be an uncertainty matroid for which   an $\A$-basis exists and $(\M^c,w^c)$ its certain weighted matroid. Then $T$ is an $\A$-basis if and only if:
	\begin{romanenumerate}
		\item $T$ contains every blue uncertain element, 
		\item $T$ avoids each red uncertain element, and
		\item The certain elements of $T$ form a minimum weight basis of  $(\M^c,w^c)$.
	\end{romanenumerate}
\end{theorem}
\begin{proof}
	$\A$-bases always contain every blue uncertain element and avoid each red uncertain element by Lemma \ref{blueAndRedUncertain}. By Lemma \ref{colorPreservation}, we can delete each red uncertain element while preserving colors. More so, from Lemma \ref{blueAndRedSafety} we conclude that $T$ is uniformly optimal after these deletions. A similar argument allows us to now contract each blue uncertain element while preserving colors in each contraction. We are only left with the certain elements of $T$ and by Lemma \ref{blueAndRedSafety} we conclude they form a minimum weight basis of $(\M^c,w^c)$.
	
	We show the converse by induction on the number of uncertain elements $k$ of $E$.
	If $k=0$, using (iii) we get that $T$ is a minimum weight basis of $(\M^c,w^c)$. Noting that $\M=\M^c$ and $\A=w^c$ we conclude that $T$ is an $(\M,\A)$-basis.
	
	If $k>0$ select any uncertain element $e \in E$. If $e$ is blue, we have from (i) that $e\in T$. As colors are preserved when contracting $e$, it follows by inductive hypothesis that $T-e$ is an $(\M/e,\A|_{E-e})$-basis, and using Lemma \ref{blueAndRedSafety}  we conclude that $T$ is an $(\M,\A)$-basis. If $e$ is red, then $e\notin T$ by (ii). We now proceed as before but deleting $e$ instead.
\end{proof}

Theorems \ref{existanceCharacterization} and \ref{UniformBasesCharacterization} allow for algorithmic implementation. We can decide if an $\A$-basis exists by checking if all elements are colored. If every element is colored, we contract every blue uncertain element, delete each red uncertain element, compute a minimum weight basis of the certain weighted matroid and output the optimal certain basis along with every blue uncertain element.  We summarize this result below and discuss its implementation in more detail in Section \ref{sec:algorithm}.
\begin{corollary}
    There is an algorithm that finds an $\A$-basis  or decides that none exists in polynomial time.
\end{corollary}

We can also use Theorem \ref{UniformBasesCharacterization} to fully characterize the matroid $\mat(\M,\A)$.

\begin{corollary}\label{matroidmat}
If $(\M,\A)$ admits an $\A$-basis, then the matroid $\mat(\M,\A)$ of all $\A$-bases is a sum of minors of $\M$. In particular, if $\M$ belongs to some minor closed class of matroids (e.g., graphic, linear, gammoid) then so does $\mat(\M,\A)$.
\end{corollary}
\begin{proof}
        Let $B$ and $R$ be the sets of blue uncertain elements and red uncertain elements of $\M$ respectively. 
Consider the set of certain weights $\{w_e\colon e\in E(\M^c)\}$ and order them increasingly $w_1 < w_2 <\dots < w_k$. Let $w_0$ be an arbitrary real number such that $w_0<w_1$.

Consider the weight function $\bar{w}\colon E\setminus R\to \mathbb R$ given by $\bar{w}_e = w_e$ if $e\in E\setminus B$, and $\bar{w}_e=w_0$ if $e\in B$. 
By Theorem \ref{UniformBasesCharacterization}, the bases of $\mat(\M,\A)$ have the form $B\cup X$, where $X$ is a minimum weight base of $(\M^c,w^c)$, where we recall that $\M^c=\M / B \setminus R$ and $w^c=\A|_{E\setminus (R\cup B)}$. Since $B$ is independent on $\M$, it is easy to see that the bases of $\mat(\M,\A)$ coincide with the minimum weight bases of $(\M\setminus R,\bar{w})$.

Define the following sets 
\begin{align*}
E_i &= \{e\in E\setminus R\colon \bar{w}_e=w_i\}, \forall i\in \{0,\dots, k\}\\
F_i&= \bigcup_{j=0}^i E_i.
\end{align*}

Standard matroid arguments (by using the greedy algorithm) show that every minimum weight basis $X$ of $(\M\setminus R,\bar{w})$ can be obtained by selecting for each $i\geq 1$, a basis $X_i$ of $\M / F_{i-1} |E_i$; the unique basis $X_0=B$ of $\M|_{E_0}=\M|_B$; and then taking $X=\bigcup_{i=0}^k X_i$.

It follows from here that $\mat(\M,\A) = \M|_{E_0} \oplus \bigoplus_{i=1}^k\M/F_{i-1}|_{E_i}$ which is a sum of minors of the original matroid.
\end{proof}

\section{Feasible Queries}\label{sec:Feasible}

\begin{definition}
	A set $F\subseteq E$ is a feasible query (or simply feasible) if no matter its revelation it guarantees the existence of a uniformly optimal basis. That is, $\forall \B \in \R(F,\A)$ there exists some $\B$-basis.
\end{definition}

Any superset of a feasible query  is also feasible. Using this, one can show that  feasible sets for a given uncertainty area function are also feasible for any revelation of a subset.

\begin{lemma}\label{feasibilityPreservedUnderRevelation}
    Let $F\subseteq E$ be feasible in $(\M,\A)$ and $X\subseteq E$. If $\B \in \R(X,\A)$, then $F$ is feasible for $(\M,\B)$.
\end{lemma}
\begin{proof}
   Since $F$ is feasible in $(\M,\A)$, so is $X\cup F$. Consider any revelation $\B'\in \R(X,\B)$. Since $\B'\in \R(X\cup F,\A)$ and $X\cup F$ is feasible, we conclude that there exists a $\B'$-basis. 
\end{proof}

Since revealing elements that are certain does not yield extra information, we also get that all minimal (for inclusion) queries only contain uncertain elements. A simple, yet strong result is that it never pays off to query a colored element. 

\begin{lemma}\label{dontRevealColored} Let $X\subseteq E$ a feasible query. If $e\in X$ is colored, then $X-e$ is a feasible query. In particular, if $X$ is a feasible query minimal for inclusion then $X$ consists only of uncertain uncolored elements.
\end{lemma}
\begin{proof}
    Let $\B\in \R(X-e,\A)$. We need to prove that  an $(\M,\B)$-basis exists. Consider any value $w_e\in \A(e)$ and the revelation $\bar{\B}\in \R(\{e\},\B)$ given by $\bar{\B}(f)=\B(f)$ for $f\in E-e$ and $\bar\B(e)=\{w_e\}$. Since $\bar\B \in \R(X,\A)$ and $X$ is feasible in $(\M,\A)$, we conclude that there exists an $(\M,\bar\B)$-basis. Now we have two cases:
If $e$ is blue in $(\M,\A)$, then, by Lemma \ref{colorPreservationRevelation} it is also blue in $(\M,\B)$ and $(\M,\bar{\B})$. By Lemma \ref{blueAndRedUncertain} the latter admits a uniformly optimal base containing $e$, say $T+e$, hence Lemma \ref{blueAndRedSafety} implies that $T$ is also a  $(\M/e,\bar{\B}|_{E-e})$-basis. But observe that $\B|_{E-e}$ is equal to $\bar{\B}|_{E-e}$. Therefore, $T$ is also an $(\M/e,\B|_{E-e})$-basis, and by Lemma \ref{blueAndRedSafety}, $T+e$ is an $(\M,\B)$-basis. The case when $e$ is red is completely analogous, using deletion instead.  We conclude that $X-e$ is a feasible query. 
\end{proof}

An exciting application of this lemma is that the set of all uncolored uncertain elements is always a feasible query set. To see this,  start with all the uncertain elements (which are a feasible query) and repeatedly remove colored elements while applying Lemma \ref{dontRevealColored}. 

By Theorem \ref{existanceCharacterization}, a set $F$ is feasible only if after its revelation all uncertain elements are colored. So, consider any uncertain element $e$ before any revelation. Intuitively, its color depends on the possible relative position between its real weight $w_e$ and the real weight of elements that could potentially span it (or cospan it). In particular, it is not hard too see that the \emph{color} of $e$ is unaffected if we reveal an element $f$ whose uncertainty area is too low (say $U_f \leq L_e$), because the real value of $f$ will be in every case at most that of $e$. A similar situation happens if the uncertainty area is too high (say $U_e\leq L_f$). A complicated thing occurs if $\A(f)$ intersects $(L_e,U_e)$, because after revealing $f$ we may still don't know the relative positions of $w_e$ and $w_f$ until we reveal $e$. Finally, if $f$ is in none of the previous situation, then after revealing it we will know for sure the relative position of $w_e, w_f$ even without revealing $e$. The previous discussion motivates the following definitions.

\begin{definition}
    For each $e\in E$ define the sets $\lo(e)$, $\mi(e)$, $\hi(e)$ and $\bo(e)$ by:
	\begin{align*}
	\lo(e) & =\{ f \in E-e : U_f \leq  L_e \}, \qquad	\hi(e) =\{ f \in E-e : U_e \leq L_f \} ,\\
	\mi(e) & =\{ f \in E-e : \A (f) \cap (L_e, U_e) \neq  \emptyset \} ,\\
	\bo(e) & =\{ f \in E\setminus \mi(e) -e : \A(f) \cap   (-\infty, L_e] \neq \emptyset \; \land \; \A(f) \cap   [U_e, \infty) \neq \emptyset \}.
	\end{align*}
\end{definition}

Note that that for $e\in E$, $F(e)=(E-e)\setminus \hi(e)$ and $F^*(e)=(E-e)\setminus \lo (e)$. Furthermore if $e$ is uncertain (i.e., $L_e< U_e$) then $E-e$ is partitioned into the sets $\lo(e), \hi(e), \mi(e)$ and $\bo(e)$.

\begin{definition}\label{matroidE}
 For each $e\in E$ denote $\M/\lo(e)\setminus \hi(e) $ by $\M'_e$
\end{definition}

\begin{remark}
	In this section we talk about different revelations simultaneously (for instance, $\A$ and $\B$). We differentiate the objects that arise this way  by using superscript denoting these dependencies. For example $F^ \A(e)$ denotes $F(e)$ with respect to the areas given by $\A$.
\end{remark}
 
The next technical lemma formalizes the idea that the elements that influence the color of an uncertain element $e$ are those in $\mi(e)$ and those in $\bo(e)$ that are not queried.

\begin{lemma}\label{uncoloredUncertainInRevelation}
Let $X \subseteq E$ and $\B\in \R(E\setminus X, \A)$ a revelation of its complement. If $e \in X$ is uncertain and uncolored in $\B$, then there exists a circuit $C$ in $\M'_e$ such that $e \in C$ and $(C-e) \cap [\mi^\A(e) \cup (X\cap \bo^\A(e))]\neq \emptyset$.
\end{lemma}

\begin{proof}
    	Consider the revelation  $\widetilde \B\in \R(\bo(e)\setminus X, \A)$ such that $\widetilde\B(f) = \B(f) (= \{B_f\})$ if $f\in \bo(e)\setminus X$ and $\widetilde\B(f)=\A(f)$ otherwise and note that $\B \in \R (E\setminus X, \widetilde \B)$. By Lemma \ref{colorPreservationRevelation}, since $e$ is uncolored in $(\M,\B)$, we get that $e$ is also uncolored in $(\M,\widetilde \B)$. Define the auxiliar sets $Y= \{ f\in \bo^\A(e) \setminus X :  B_f \leq L_e \}$ and $\overline{Y}=\{ f\in \bo^\A(e) \setminus X :  B_f > L_e \}$. Note that $\lo^{\widetilde \B} (e) =\lo^{\A} (e) \cup Y; \hi^{\widetilde \B}(e) = \hi^{\A} (e) \cup \overline{Y}; \mi^{\widetilde \B}(e) = \mi^{\A} (e);$ and $\bo^{\widetilde \B} (e) = X \cap \bo^{\A}(e)$.
	Since $e$ is uncertain and it is not blue nor red in $(\M,{\widetilde B})$ we get:
	$e\in \spn [\lo^{\widetilde \B}(e)\cup \mi^{\widetilde \B}(e)\cup \bo^{\widetilde \B}(e)] = \spn[\lo^\A(e)\cup \mi^\A(e)\cup Y \cup (X\cap \bo^\A(e))],$ and $
\!\!\!\!\!e\in \cospn [\hi^{\widetilde \B}(e)\cup \mi^{\widetilde \B}(e)\cup \bo^{\widetilde \B}(e)] = \cospn[\hi^\A(e)\cup \mi^\A(e)\cup \overline Y \cup (X\cap \bo^\A(e))] .$
Using that $e\in \cospn[Q]$ implies $e \notin \spn [(E-e)\setminus Q]$ for any $Q\subseteq E$, we conclude: \begin{equation}e\in \spn [\lo^\A(e) \cup \mi^\A(e) \cup Y \cup (X \cap \bo^\A(e))] \setminus \spn [\lo^\A(e) \cup Y]. \label{eqn1}\end{equation}
	
    From \eqref{eqn1}, $e \in \spn [\lo^\A(e) \cup \mi^\A(e) \cup Y \cup (X \cap \bo^\A(e))] \setminus \lo^\A(e)=\spn_{M'_e} [\mi^\A(e) \cup Y \cup (X \cap \bo^\A(e))]$, where the equality follows from properties of contraction and deletion.  Therefore there is a circuit $C$ in $\M'_e$ such that $e\in C$ and $C-e \subseteq \mi^\A(e) \cup Y \cup X\cap \bo^\A(e)$. If $(C-e)\cap [\mi^\A(e)\cup (X \cap \bo^\A(e))]= \emptyset$ , we would have that $C-e \subseteq Y$, implying that $e\in \spn_{\M'_e}Y=\spn[\lo^\A(e)\cup Y]\setminus \lo^A(e)$ which contradicts \eqref{eqn1}.
\end{proof}

The previous lemma is useful to characterize the sets that intersect every feasible set.
\begin{definition}\label{witnessSet}
	A set $X\subseteq E$ is a witness set if it intersects every feasible set.
\end{definition}

Witness sets have been studied before in the context of online adaptive algorithms for MST and matroids. Since the definition of feasible is slightly different in that settings (they are feasible for the verification problem, in which one knows the real values a priori), these witness sets are also different from ours. 

\begin{lemma}\label{witnessSetsCharacterisation}
	Let $X\subseteq E$. The following statements are equivalent:
	\begin{romanenumerate}
		\item $X$ is a witness set.
		\item There exists an uncertain element $e\in X$ and a circuit $C$ in $\M_e'$ such that $e \in C$ and $C \cap [\mi (e) \cup (X \cap \bo(e))] \neq \emptyset$. 
	\end{romanenumerate}
\end{lemma}
\begin{proof}
Let $X$ be a witness set. Since $E\setminus X$ isn't feasible there is a revelation $\B \in \R(E\setminus X,\A)$ such that there is no $(\M,\B)$-basis and by Theorem \ref{existanceCharacterization} we must have an element $e$ that is uncertain and uncolored in $(\M,\B)$. Note that $e \in X$ as every element in $E\setminus X$ is certain in $(\M,\B)$. We conclude by using Lemma \ref{uncoloredUncertainInRevelation} on $e$.
	
	For the converse, set $Y = [\bo(e) \setminus X ]\cap C$, $\overline{Y}= [\bo(e)\setminus X] \setminus C$ and note that
	$e \in \spn_{\M_e'} (C-e)= \spn_{\M_e'} [(C \cap \mi(e))  \cup (C \cap (\bo(e) \setminus X)) \cup (C \cap \bo(e) \cap X)]  \subseteq \spn_{\M_e'}  [\mi(e) \cup Y  \cup  (X \cap \bo(e)) ].$
	
	If $e\in \spn_{\M'_e} (Y)$  there would be a circuit $D$ in $ \M'_e$ such that $D-e \subseteq Y\subseteq C$, but since $C \cap [\mi(e)\cup (X \cap \bo(e))]\neq \emptyset$ we get $D\subsetneq C$ which contradicts the minimality of $C$ as circuit. Then,
	$e \in \spn_{\M'_e} [\mi(e) \cup Y \cup (X \cap \bo(e))] \setminus \spn_{\M'_e} Y$, which is included in 
	$\spn [\lo(e) \cup \mi(e) \cup Y \cup (X \cap \bo(e))] \setminus \spn [\lo(e)\cup Y].$
	
	As $e \notin \hi(e)\cup \mi(e)\cup \overline Y \cup (X\cap \bo(e))$ we conclude that $e\in \spn [\lo(e) \cup \mi(e) \cup Y \cup (X\cap \bo(e))]$ and $e \in \cospn [\hi(e) \cup \mi(e) \cup \overline Y \cup (X\cap \bo(e))]$.
	
	Choose two revelations $w^+,w^- \in \R(E,\A)$ as follows:
	\begin{align*}
	\!\!\!\!\!\!\!w_f^+ & \in \begin{cases}
	(L_e,U_e)\cap \A(f) & \text{if $f\in \mi(e)$},\\
	(-\infty, L_e]\cap \A(f) & \text{if $f \in Y$ or}\\
	&\text{$f  \in X \cap \bo(e)$}, \\
	[U_e,\infty)\cap \A(f) & \text{if $f \in \overline Y$}, \\
	\A(e) \cap (L_e, U_e] & \text{if $f=e$}.
	\end{cases} &\!\!\!
	w_f^- & \in \begin{cases}
	 [U_e,\infty)\cap \A(f) & \text{if $f \in X \cap \bo(e)$}, \\
	\A(e) \cap [L_e,U_e) & \text{if $f=e$}, \\
	\{ w_f^+ \}   & \text{otherwise.}
	\end{cases} 
	\end{align*}
	Note that $w^+$ and $w^-$ only differ on $X \cap \bo(e)$ and $e$. As $e \in \spn [\lo(e) \cup \mi(e) \cup Y \cup (X \cap \bo(e))] $ it is the unique heaviest element in a circuit in $(\M,w^+)$, therefore it is in no $(\M,w^+)$-basis. Similarly, since $e\in \cospn [\hi(e) \cup \mi(e) \cup \overline Y \cup (X \cap \bo(e))] $ it is the unique lightest element in a cocircuit in $(\M,w^-)$, hence it is in every $(\M,w^-)$-basis. To conclude suppose that there is a feasible query set $F$ such that $X\cap F =\emptyset$, we then pick the  revelation $\B \in \R(F,\A)$ such that $\B(f)=\{w_f^+\}=\{w_f^-\}$ if $f\in F$, and $\B(f)=\A(f)$ otherwise.
	
	Select any $\B$-basis $T$. As $w^+,w^- \in \R(E,\B)$, $T$ is both a $w^+$-basis and a $w^-$-basis. By the previous paragraph, this implies that $e\not\in T$ and $e\in T$ which is a contradiction.
\end{proof}

\begin{lemma}\label{blocker}
\!\!The minimal feasible queries are the sets intersecting every witness set.
\end{lemma}
\begin{proof}
Recall that a \emph{clutter} is a family of sets such that no one is contained in another. The \emph{blocker} of a clutter is the clutter of all minimal sets that intersect the first one. Thus, the minimal witness sets are the blocker of the minimal feasible queries. A basic result in packing and covering theory (see e.g., \cite[Theorem 77.1]{Schrijver}) states that the blocker of the blocker of a clutter is again the original clutter. The lemma follows from this fact.
\end{proof}

As we see below, minimal witness sets cannot be large: they can have at most 2 elements.

\begin{corollary}
	\label{existsCriticalPair}
	Let $X$ be a witness set such that $|X|\geq 2$. Then, there exists distinct $e,f \in X$ such that $\{e,f\}$ is a witness set.
\end{corollary}
\begin{proof}
	If $X\cap F \neq \emptyset$ for every feasible query set $F$ then, by Lemma \ref{witnessSetsCharacterisation}, there exists an uncertain $e \in X$, a circuit $C$ in $ \M'_e$ such that $e \in C$ and $C\cap [\mi(e) \cup (X \cap \bo(e))]\neq \emptyset$. 
	
	If $C \cap (X \cap \bo(e))=\emptyset$, then $C \cap \mi(e)\neq \emptyset$ and by Lemma \ref{witnessSetsCharacterisation} we have that $\{e\} \cap F\neq \emptyset$ for every $F$ feasible query set. Picking any $f \in X-e$ we conclude that $\{e,f\}\cap F \neq \emptyset$ for every feasible query set $F$. If $C \cap (X \cap \bo(e))\neq \emptyset$, select any $g \in C \cap (X \cap \bo(e))$. Once again, Lemma \ref{witnessSetsCharacterisation} lets us conclude that $\{e,g\} \cap F \neq \emptyset$ for every $F$ feasible query set.    
\end{proof}

\begin{definition}\label{criticalRelation}
Let $\core = \{ e \in E \colon \{e\} \text{ is a witness set} \}$ and $\overline{\core}=E\setminus \core$. Define the graph $G^{\text{wit}}=(\overline\core, E^{\text{wit}})$ where $ef \in E^{\text{wit}}$ if $\{e,f\}$ is a witness set.
\end{definition}

By Lemma \ref{blocker}, minimal feasible sets are exactly those sets containing all elements in $\core$ together with a vertex cover of $G^{\text{wit}}$. The following clean characterization of $\core$ follows directly from Lemma \ref{witnessSetsCharacterisation}.

\begin{lemma}\label{coreCharacterization}
  Let $e$ be an uncertain element.  $e \in \core$ if and only if there is a circuit $C$ in $\M'_e$ such that $e\in C$ and $C\cap \mi(e)\neq\emptyset$.
\end{lemma}

We can turn the previous lemma into an algorithm that computes  \core. In order to do this we compute the connected component\footnote{Recall that $e$ is connected to $f$ in a matroid if and only if there is a circuit that contains $e$ and $f$. A connected component is an equivalence class of this equivalence relation (see, eg. \cite[Section 4.1]{Oxley}).} of $\M'_e$ that contains $e$ and check if it has non-empty intersection with $\mi(e)$. We compute connected components with an algorithm due to Krogdhal \cite{Krogdahl77} that takes polynomial time. Therefore, the previous procedure also takes polynomial time. In what follows we  show that $G^{\text{wit}}$ has a very nice structure. 

\begin{lemma}
	\label{criticalLemma}
	\begin{romanenumerate}
		\item Let $e,f \in 	\overline{\core}$ distinct. $ef\in E^{\text{wit}}$ if and only if $\A(e)=\A(f)=\{L_e,U_e\}$ with $L_e<U_e$ and there is a circuit $C$ of $\M_e'$ such that $e, f \in C$.
		\item The connected components of the graph $G^{\text{wit}}$ are cliques.
	\end{romanenumerate}
\end{lemma}
\begin{proof}
	\begin{romanenumerate}
		\item Suppose  $ef\in E^{\text{wit}}$. Since $\{e,f\}$ is a witness set, $E\setminus \{e,f\}$ is not feasible. Let $\B \in \R(E\setminus \{e,f\},\A)$ be a revelation without $\B$-basis. Since $f \notin \core$, we have that $E-f$ is feasible in $(\M,\A)$ and by Lemma \ref{feasibilityPreservedUnderRevelation} it is also feasible in $(\M,\B)$. If $e$ was colored in $(\M,\B)$ then, by Lemma \ref{dontRevealColored} we would conclude that $(E-f)-e = E\setminus \{e,f\}$ is feasible in $(\M,\B)$. But all elements in $E\setminus\{e,f\}$ are already certain in $(\M,\B)$, from which we deduce that $\emptyset$ is feasible in $(\M,\B)$, which is a contradiction.
		
		We conclude that $e$ is uncolored in $(\M,\B)$.
		From this, we can use Lemma \ref{uncoloredUncertainInRevelation} to obtain a circuit $C$ in $\M'_e$ such that $e \in C$ and $C \cap [\mi^\A(e) \cup (\{e,f\}\cap \bo^\A(e))]\neq \emptyset$.	Since $e\not \in \core$, we have by Lemma \ref{coreCharacterization} that $C \cap \mi(e) = \emptyset$.  Then $C \cap \{e,f\}  \cap \bo(e)\neq \emptyset$, consequently $e,f \in C$ and $f \in \bo(e)$.
		We can now use the same argument for $f$, concluding that $e \in \bo(f)$. The only way for $e\in \bo(f)$ and $f\in \bo(e)$ to occur at the same time is that  $\A(e)=\A(f)=\{L_e,U_e\}$. Finally, since witness sets do not contain elements that are certain (this follows since minimal feasible sets only have uncertain elements, hence by removing certain elements from a witness set it would still intersect every feasible set) we must have $L_e<U_e$.
		
		We now prove the converse. As $e$ is an uncertain element such that there is a circuit $C $ in $\M_e'$ and $f \in C \cap \{e,f\} \cap \bo(e)  $ by Lemma \ref{witnessSetsCharacterisation} we conclude that $ef \in E^{\text{wit}}$.
		
		\item We only need to show that whenever $ef, fg \in E^{\text{wit}}$ we also have $eg \in E^{\text{wit}}$. Suppose that $ef, fg \in E^{\text{wit}}$. By the previous item we have $\A(e)=\A(f)=\A(g)=\{L_e,U_e\}$, in particular $\M_e'=\M_f'=\M_g'\doteq \M'$. The previous item also allows us to conclude that there are two circuits $C^1, C^2$ in $\M'$ such that $e,f \in C^1$ and $f,g\in C^2$. Then $e,f$ and $g$ are in the same matroid connected component in $\M'$. Therefore there is a circuit $C^3$ in $\M'$ such that $e,g \in C^3$ and using the previous item we conclude that $eg\in E^{\text{wit}}$.\qedhere
		\end{romanenumerate}	\end{proof}
		
		We can test if $ef\in E^{\text{wit}}$ similarly to how we computed $\core$: we start by considering elements with areas of size two (by checking if $(L_e,U_e)\cap \A(e)=\emptyset$ for every element $e$). If $e$ and $f$ have the same two-element uncertainty area, we check if they belong to the same connected component in $\M'_e=\M'_f$ using Krogdhal's algorithm \cite{Krogdahl77}. 
		
		\begin{theorem}\label{feasibleMinSize} $X\subseteq E$ is a minimal feasible query if and only  if $\core \subseteq X$ and $X$ intersects all but one element in each connected component of $G^{\text{wit}}$.
	\end{theorem}
\begin{proof}
By Lemma \ref{blocker} and the definition of $G^{\text{wit}}$ the minimal feasible queries $X$ satisfies that $\core \subseteq X$ and $X\cap \overline{\core}$ is a minimal vertex cover of $G^{\text{wit}}$. Since every connected component is a clique, the minimal vertex covers  of $G^{\text{wit}}$ are exactly those sets containing all but one element in each connected component. 
\end{proof}

\begin{corollary}\label{algorithmMCFQ}
      Let $c: E \to \mathbb R$ be any cost function. We can compute a minimum-cost feasible query in polynomial time.
\end{corollary}
\begin{proof}
Computing $\core$ and $G^{\text{wit}}$ can be done in polynomial time and polynomial number of calls to the independence oracle of $\M$ or to minors of $\M$ (for example, $\M'_e$), since the oracle of independence of any minor of $\M$ can also be implemented using a polynomial number of calls to the oracle of $\M$.
One can compute a minimum-cost minimal size feasible query $F$ by simply returning a set  containing $\core$ and all but the most expensive element from each connected component of $G^{\text{wit}}$. If we allow negative costs, then the minimum-cost feasible query is $F$ together with all the negative cost elements outside $F$. An efficient implementation is discussed in Section \ref{sec:algorithm}
\end{proof}

We finish this section with two special cases of the last theorem.

\paragraph*{Interval Uncertainty Areas}
\begin{corollary}\label{feasibleInterval}
Suppose that $\A(e)$ is an interval (of any type: open, closed, semiopen, trivial) for every element $e\in E$. Then the set $S$ of all uncolored uncertain elements is the only minimum-size feasible query.
\end{corollary}
\begin{proof}
Lemma \ref{criticalLemma} implies that the graph $G^{\text{wit}}$ has no edges. Applying Theorem \ref{feasibleMinSize} we conclude that the only minimum-size feasible query is $\core$. We now prove that $\core=S$. Since $S$ is feasible,  we have  that $\core\subseteq S$. For the other inclusion, set $X=E$, $\B=\A$ on Lemma \ref{uncoloredUncertainInRevelation} to conclude that for every element $e\in S$, there is a circuit $C$ in $\M_e'$ such that $e\in C$ and $C\cap [\mi(e) \cup \bo(e)]\neq \emptyset$. Since areas are all intervals and $e$ is uncertain, $\bo(e)=\emptyset$. Therefore $C\cap \mi(e)\neq \emptyset$, and Lemma \ref{coreCharacterization} allows us to conclude that $e \in \core$.
\end{proof}
\paragraph*{MST with 0-1 Uncertainty Areas}
\begin{corollary}\label{feasible01}
Let $G$ be a connected graph such that for every $e\in E(G)$, $\A(e)=\{0,1\}$. A set $F\subseteq E(G)$ is a  minimal feasible query for the MST problem if and only if $F$ contains all but one edge from  each 2-connected component of $G$. 
\end{corollary}
\begin{proof}
    Note that for each $e$, $\mi(e)=\emptyset$. Therefore, by Lemma \ref{coreCharacterization}, $\core=\emptyset$. Moreover, for every $e$, $\M'_e=\M$. Hence, by Lemma \ref{witnessSetsCharacterisation}, $ef \in E^{\text{wit}}$ if and only if there is a cycle $C$ in the graph $G$ containing both. In other words, the connected components of  $G^{\text{wit}}$ correspond exactly to the edge-sets of the blocks (2-connected components) of $G$. The result then follows from Theorem \ref{feasibleMinSize} 
\end{proof}
\section{Algorithmic Implementations for coloring, for finding $\A$-bases and for finding minimum cost feasible queries}\label{sec:algorithm}

\subsection{Coloring algorithms}

Recall that for any $Q \subseteq E$ and $e \notin Q$, one can decide if $e \in \spn Q$ by: 
\begin{romanenumerate}
\item Selecting a basis of $Q$ (for example, via the greedy algorithm with $O(|Q|)$ calls to the independence oracle),
\item Checking if $Q+e$ is not independent.
\end{romanenumerate}

Hence, after sorting the elements by infima and suprema, we can decide if an element is blue by using Lemma \ref{blueAndRedCharacterization}. We simply check for each $e$ if $e\notin \spn F(e)$. To test for redness we recall that $\cospn(X)=X \cup \{e\in E : e\notin \spn[(E-e)\setminus X] \}$. Since $e\notin F^ *(e)$, we just check if $e \in \spn ( (E-e) \setminus F^*(e))$. 

By the previous paragraph, we can color all elements in $O(|E|^2)$ time. Additionally, we can use Theorem \ref{UniformBasesCharacterization} to provide an algorithm that finds uniformly optimal bases in $O(|E|^2)$ time.

\subsection{A faster algorithm for finding $\A$-bases.}

 Lemma \ref{closureLemma} shows that the uncertainty matroids $(\M,\cl \A)$ and $(\M,\A)$ have the same colors. Aditionally, they have the same certain elements. By Theorem \ref{UniformBasesCharacterization} we know that uniformly optimal bases only depend on colors and which elements are certain. Therefore, $(\M, \cl \A)$ and $(\M,\A)$ have the same set of uniformly optimal bases. Consequently, we can use any algorithm for closed intervals in the general case, simply by replacing $\A$ with $\cl A$. In particular, we can use the following regret-based algorithm by Kasperski and Zielinski \cite{KasperskiZ06}, which relies only on executing 2 greedy algorithms with appropriate weight functions.
 
 	\begin{algorithm}[H]
	\caption{\texttt{Regret based UOB algorithm}}
	\begin{algorithmic}[1]
		\Statex \textbf{Input:} $\langle \M,\A\rangle$ where $(\M,\A)$ is an uncertainty matroid.
		\Statex \textbf{Output:} An $\A$-basis if one exists, otherwise \texttt{FALSE}.
		\For{$e \in E$}
			\State $w_e \gets \frac{L_e+U_e}{2}$
		\EndFor
		\State $T \gets \texttt{greedy}(\M,w)$;
		\For{$e \in T$}
			\State $w'_e \gets U_e$
		\EndFor
		\For{$e \notin T$}
		\State $w'_e \gets L_e$
		\EndFor
		\State $T'\gets \texttt{greedy}(\M,w')$;
		\If{$w'(T')<w'(T) $}
			\State\Return \texttt{FALSE};
		\Else
			\State\Return $T$;
		\EndIf
	\end{algorithmic}
\end{algorithm}		

This provides an algorithm faster than that of the previous subsection, taking only $O(|E| \log |E)$ time.

\subsection{An algorithm for finding minimum cost feasible queries.}

By our assumption on the uncertainty areas, testing membership of an element $f$ in $\lo(e), \hi(e), \bo(e),$ and $\mi(e)$ can be done in constant time. Note that after precomputing all four sets in $O(|E|)$ time, and after computing a base of $\lo(e)$ in $\M$ (also in $O(|E|)$ time), testing independence in the matroid $\M'_e=\M/\lo(e)\setminus \hi(e)$ can be done with a single oracle call to $\M$.

Testing if a given element $f$ is in the same connected component than $g$ in some matroid $\M$ in linear time is simple: compute first an $\M$ base $T$ using the greedy algorithm starting from $g$ (so that $g\in T)$. $f$ is in the same connected component as $g$ if and only if $T+f-g$ is independent in $\M$ and $T+f$ is not.  A naive way to compute all connected components would be to perform this procedure $O(|E|^2)$-times (for every pair $(f,g)$), to get an $O(|E|^3)$-time algorithm. We can do better by using Krogdahl's algorithm \cite{Krogdahl77}. This algorithm computes any single base $T$ and then computes the dependency bipartite graph associated to $T$. This is the graph $H$ with sides $T$ and $E\setminus T$, where $(f,g)\in T\times (E\setminus T)$ is an edge if and only if $T+g-f$ is independent. It turns out that the (graph) connected components of $H$ correspond exactly to the (matroid) connected components of $\M$. Note that $H$ can be constructed by only using the greedy algorithm once (in time $O(|E|)$), and then using an extra $O(|E|^2)$ independence oracle calls, for a total of $O(|E|^2)$ time.

We can use the previous discussion, to provide the following implementation of Corollary \ref{algorithmMCFQ}.

	\begin{algorithm}[H]
	\caption{\texttt{MCFQS algorithm}}
	\begin{algorithmic}[1]
		\Statex \textbf{Input:} $\langle \M,\A,c\rangle$ where $(\M,\A)$ is an uncertainty matroid and $c: E\to \mathbb {R}$ a cost function.
		\Statex \textbf{Output:} A minimum cost feasible query.
		\State $Q\gets \emptyset$;
		\For{$e \in E(M)$}
			\State Compute $\lo(e)$, $\mi(e)$, $\hi(e)$ and $\bo(e)$;
		\EndFor
		\For{$e \in E(M)$}
		    \State $K \gets \{ f \in E(M '_e)  :   \text{$f$ and $e$ are connected in $\M_e'$}\}$;
		    \If{$K \cap \mi(e)=\emptyset$}
		\State $Q \gets Q+e$;
		\EndIf
		\EndFor
		\State $E' \gets \{e \in  E\setminus Q : |A_e|=2 \}$;
				\While{$E'\neq \emptyset $}
			\State Select any $e \in E'$;
			\State $\Gamma \gets \{ f \in E(\M '_e)  :   \text{$f$ and $e$ are connected in $\M_e'$ and $\A(e)=\A(f)$}\}$;
			\State Choose any $e_\Gamma \in \argmax \{ c_f : f\in \Gamma \}$;
			\State $Q \gets Q \cup (\Gamma - e_\Gamma)$; 
			\State $E'\gets E'\setminus \Gamma$;
		\EndWhile
		\State \Return $Q \cup \{e \in E : c_e < 0 \}$;
	\end{algorithmic}
\end{algorithm}	

The second for loop executes Krogdhal's algorithm once for each element in order to compute the connected components in $\M'_e$, taking $O(|E|^3)$ time. Note that after this loop $Q$ is exactly $\core$ by Lemma \ref{coreCharacterization}. The third loop executes Krogdhal's algorithm once for each element picked in the while loop, in order to compute its connected component in $\M'_e$, taking $O(|E|^3)$-time. Note that at each iteration $\Gamma$ is exactly a connected component of $G^\texttt{wit}$ by Lemma \ref{criticalLemma}. Hence, the algorithm terminates in $O(|E|^3)$ time and correctly outputs a minimum cost feasible query set.

\bibliography{Uncertainty}

\appendix

\section{Direct proofs for red cases}

As stated in the main text, many lemmas concerning red elements follows directly from the blue case via a duality argument. For completeness, we also include direct proofs that do not rely on this fact.

\begin{proof}
[Proof of Lemma \ref{blueAndRedCharacterization}, Red case]
    	For the direct implication, define $K= \min\limits_{f \in F^*(e)} U_f - L_e>0$. For each $f\in E-e$ we can choose $\varepsilon_f\in [0,K/2)$ such that $U_f -\varepsilon_f \in \A (f)$ similarly we can pick $\varepsilon_e \in [0,K/2)$ such that $L_e+\varepsilon_e \in \A(e)$. We then consider the following realization $w^* \in \R(E,\A)$:
	\[
	w^*_f = \begin{cases}   L_e + \varepsilon_e  & \text{ if $f=e$,} \\
	U_f - \varepsilon_f  & \text{ if $f\neq e$.}
	\end{cases}
	\]  
	
	Note that if $f \in F^*(e)$ then $w^*_f>w^*_e$, indeed:
	\[
	w^*_f = U_f - \varepsilon_f > U_f - \frac{K}{2} \geq U_f - \frac{U_f-L_e}{2}= L_e + \frac{U_f-L_e}{2} \geq  L_e + \frac{K}{2} > U_e -\varepsilon_e = w^*_e
	\]
	Suppose now that $e \in \cospn F^*(e)$, then there exists a cocircuit $C^* \subseteq F^*(e)+e$ such that $e$ is the lightest element of $C^*$. This implies that $e$ is in every $w$-basis, contradicting that $e$ was red.
	
Let us suppose that $e$ is not red, then there exists a realization $w \in \R(E,\A)$ such that $e$ is in every $w$-basis. Choose $T$ to be any $w$-basis and $C^*$ the fundamental circuit of $E\setminus T+e$, as $e \notin \spn F^*(e) $ we have that $C^*-e \not\subseteq F^*(e)$. Select any $f\in C^*\setminus F^*(e)$, then:
	\[
	w_e \geq L_e \geq U_f \geq w_f
	\]
	We conclude that $T+f-e$ is a $w$-basis that avoids $e$.
\end{proof}

\begin{proof}[Proof of Lemma \ref{blueAndRedUncertain}, Red case]
    Since $T$ is a $w$-basis for every $w \in \R(E,\A)$ it is clear that if $e \notin T$, then $e$ is red.
    
    For the other implication, assume that $e\in T$. Consider $C^*$ the fundamental cocircuit of $E\setminus T+e$, as $e$ is red, we have that $(C^*-e)\not\subseteq F^*(e)$. Considering that $e$ is uncertain there exists a realization $w\in \R(E,\A)$ such that $w_e>L_e$, then selecting $f \in C^* \setminus F^*(e)$ we have:
	\[
	w_e > L_e \geq U_f \geq w_f,
	\]  
	This implies that $w(E\setminus T+e-f)>w(E\setminus T)$ and $w(T-e+f)<w(T)$, which contradicts the fact that $T$ is an $\A$-basis.
\end{proof}

\begin{proof}[Proof of Lemma \ref{blueAndRedSafety}, Red case]
For the direct implication observe that, as $e$ is red, there is a $w$-basis $T^-$ that avoids $e$. Furthermore, $T^-$ is a basis of $\M-e$. Note that $|T^-|=|T|$, then $T$ is also a basis of $\M$. Let us suppose that $T$ is not an $(\M,\A)$-basis. Then, there exists a basis $T'$ of $\M$ and a realization $w \in \R(E,\A)$ such that $w(T')<w(T)$. As $e$ is red in $(\M,\A)$ there exists $T''$, an $(\M,w)-$basis, such that $e \notin T '' $. So, $w(T'')=w(T')$. Notice that $T''$ is a basis of $\M-e$. Since $w(T''-e)<w(T)$ this contradicts the fact that $T$ is $(\M-e, \A|_{A-e})$-basis.

For the converse, it is clear that $T$ is a basis of $\M-e$. Suppose that $T$ is not an $(\M-e,\A|_{A-e})$-basis. Then, we have a basis $T'$  of $\M-e$ and a realization $w \in \R(E-e, \A|_{A-e})$ such that $w(T')<w(T)$. Extending $w$ to $\hat w \in \R(E,\A)$ by selecting any $\hat w_e \in \A(e)$ we have that $\hat w (T')< \hat w(T)$. Since $|T'|=|T|$ we have that $T'$ is a basis of $\M$ which contradicts that $T$ is an $(\M,\A)$-basis.
\end{proof}

\begin{proof}[Proof of Lemma \ref{colorPreservation}, Red case]
    	Define $w,w^*:E \to \mathbb R$ by:
    		\[
	w_x = \begin{cases}
		U_f & \text{if $x=f$,}\\
		L_x & \text{if $x\neq f$.}
	\end{cases} \qquad 	w^*_x = \begin{cases}
	L_f & \text{if $x=f$,}\\
	U_x & \text{if $x\neq f$.}
	\end{cases}
	\]
	We begin by considering $f$ red. Since $f$ is red in $(\M,\cl \A)$, there exists some $(\M,w^*)$-basis $T$ such that $f \notin T$. We first prove that $e\notin T$, suppose not, we take $C^*$ the fundamental cocircuit of $E\setminus T+e$. As $e$ is red, we have that $e \notin \cospn F^*(e)$ and $(C-e)\not \subseteq F(e)$. Then, we can select $g \in C^* \not\subseteq F^*(e)$ such that $g\neq e$. Since $e$ is uncertain: 
	\[w^*_e=U_e > L_e \geq U_g \geq w^*_g ,\]  
	it follows that $w^*(E\setminus T-g+e)>w^*(E\setminus T)$ and $w^*(T-e+g)<w^*(T)$, which contradicts the fact that $T$ is an $(\M,w^*)$-basis. Therefore $E\setminus T$ is a maximum weight cobasis (w.r.t $(\M-e,w^*)$) such that $f\in T$. Hence, \[f\notin \cospn_{\M-e} \{g \in E-e : w^*_g > w^*_f \}=\cospn_{\M-e} (F^*(f)-e).\] We conclude that $f$ is red in $(\M-e, \A|_{A-e})$ by using Lemma \ref{blueAndRedCharacterization}.
	
		 We now turn to the case when $f$ is blue. As $f$ is blue in $(\M,\cl \A)$, there exists some $(\M,w)$-basis $T$ such that $f \in T$. If $e\in T$ then $T$ is an $(\M-e,w|_{E-e})$-basis and $f\notin \spn_{\M-e} (\{g \in E-e : w_g < w_f \})=\spn_{\M-e} (F(f)-e)$. By Lemma \ref{blueAndRedCharacterization} we conclude that $f$ is blue in $(\M,\A)$. Therefore, we are only left with the case $e\in T$. Since $e$ is red in $(\M,\cl \A)$, there exists some $(\M,w)$-basis $T'$ such that $e\notin T'$. We can find $e ' \in T'\setminus T$  such that $T-e+e'$ is an $(\M,w)$-basis, from here we can proceed as before but working with $T-e+e'$ instead of $T$.

\end{proof}

\end{document}